\documentclass{article}
\usepackage{amsthm}
\usepackage{latexsym}
\usepackage{amsmath}
\usepackage{amssymb,mathrsfs}
\usepackage{paralist}
\usepackage{graphics} 
\usepackage{graphicx}
\usepackage{epsfig} 

\newcommand{\bbZ}{\mathbb{Z}}
\newcommand{\bbN}{\mathbb{N}}

\newtheorem{thm}{Theorem}[section]

\newtheorem{remark}{Remark}

\usepackage{doublespace}

\begin{document}
%
\title{Batch Arrival Multiserver Queue with Setup Time}

\author{Tuan Phung-Duc \\
Department of Mathematical and Computing Sciences\\
 Tokyo Institute of Technology \\
Ookayama, Meguro-ku, Tokyo 152-8552, Japan\\
Email: tuan@is.titech.ac.jp}


%


\maketitle

\begin{abstract}
Queues with setup time are extensively studied because they have application in performance evaluation of power-saving data centers. 
In a data center, there are a huge number of servers which consume a large amount of energy. In the current technology, an idle server still consumes about 60\% of its peak processing a job. Thus, the only way to save energy is to turn off servers which are not processing a job. However, when there are some waiting jobs, we have to turn on the OFF servers. A server needs some setup time to be active during which it consumes energy but cannot process a job. Therefore, there exists a trade-off between power consumption and delay performance. Gandhi et al. \cite{Gandhi10a,Gandhi10} analyze this trade-off using an M/M/$c$ queue with staggered setup (one server in setup at a time). In this paper, using an alternative approach, we obtain generating functions for the joint stationary distribution of the number of active servers and that of jobs in the system for a more general model with batch arrivals and state-dependent setup time. We further obtain moments for the queue size.  Numerical results reveal that keeping the same traffic intensity, the mean power consumption decreases with the mean batch size for the case of fixed batch size. One of the main theoretical contribution is a new conditional decomposition formula showing that the number of waiting customers under the condition that all servers are busy can be decomposed to the sum of two independent random variables where the first is the same quantity in the corresponding model without setup time while the second is the number of waiting customers before an arbitrary customer. 
\end{abstract}


%

\section{Introduction}
Cloud computing is a new paradigm where companies make money by providing computing service through the Internet. In cloud computing, users buy software and hardware resources from a provider and access to these resources through the Internet so they do not have to install and maintain by themselves. The core part of cloud computing is a data center where there are a huge number of servers. The key issue for the management of data centers is to minimize the power consumption while keeping an acceptable service level for customers~\cite{Barroso07,Chen05,Greenberg08,Mazzucco12,Meisner09,phungduc14,Schwartz12}. It is reported that under the current technology an idle server still consumes about 60\% of its peak processing jobs~\cite{Barroso07}. Thus, the only way to save power is to turn off idle servers. However, if the workload increases, OFF servers should be turned on to serve waiting customers. Servers need some setup time during which they consume energy but cannot process jobs. Therefore, customers have to wait a longer time in comparison with the case where the servers are always ON.

Although queues with setup time have been extensively investigated in the literature, most of papers deal with single server models~\cite{Takagi90,Bischof01,Choudhury98,Choudhury00} where the service time follows a general distribution. Artalejo et al.~\cite{Artalejo05} present a throughout analysis for multiserver queues with setup time where the authors consider the case in which at most one server can be in setup mode at a time. This policy is referred to as staggered setup in~\cite{Gandhi10}. It is pointed out in~\cite{Artalejo05} that the model belongs to a QBD class for which the rate matrix is explicitly obtainable. By solving difference equations, Artalejo et al.~\cite{Artalejo05} derive an analytical solution where the stationary distribution is recursively obtained without any approximation. Recently,  motivated by applications in data centers, multiserver queues with setup time have been extensively investigated in the literature. In particular, Gandhi et al.~\cite{Gandhi10,Gandhi10b,gandhi11,Gandhi11b,Gandhi13} analyze multiserver queues with setup time. They consider the M/M/$c$ system with staggered setup and derive some closed form approximations for the ON-OFF policy where the number of servers in the setup mode at a time is not limited. Gandhi et al.~\cite{gandhi11}  extend their analysis to the case where a free server waits for a while before shutdown. As a related model, Tian et al.~\cite{Tian99} consider M/M/$c$ model with vacation where after a service an idle server leaves for an exponentially distributed vacation. 

In all the work on multiserver queue mentioned above, customers (jobs) are assumed to arrive individually according to a Poisson process. However, in cloud computing a big task might be divided into multiple subtasks to process in parallel~\cite{Dean08}. 
This motivates us to consider a multiserver queueing system with state-dependent setup time under batch arrival settings. 
In this paper, using a generating function approach, we derive a clear solution for all the partial generating functions for the joint stationary distribution of the number of active servers and that of customers in the system. The generating functions are obtained using recursive formulae. Special cases of our model conform to the models in~\cite{Artalejo05,Gandhi10,Tian99,phungduc14b}. Furthermore, we derive a recursion which allows to calculate all the moments of the queue length. Numerical results are presented to show the affect of batch arrivals on the performance of the system. Furthermore, we present a method to derive the waiting time distribution. We also present some variants which can be analyzed by adapting the methodologies presented in this paper. One of the most important theoretical contribution is that we prove a conditional decomposition property showing that the queue length under the conditional that all the servers are busy can be decomposed into the sum of two independent random variables with clear physical meaning. 

The rest of our paper is organized as follows. First we present the model in Section~\ref{model:sec}. Section~\ref{analysis:sec} is devoted to the detailed analysis where we derive the partial generating functions and the joint stationary distribution. Section~\ref{waiting_time_distribution:sec} briefly presents the method to compute the waiting time distribution while Section~\ref{variant models} shows some variant models that could be analyzed by the methodology of this paper. In Section~\ref{decomposition}, we discuss the conditional decomposition property for the queue length. Section~\ref{waiting_time_distribution:sec} briefly shows the derivation of the waiting time distribution while in Section~\ref{variant models}, some variants are demonstrated. In Section~\ref{numerical:sec}, we provide extensive numerical results to show the performance of the system. Finally, concluding remarks are presented in Section~\ref{concluding_remark:sec}.

\section{Model}\label{model:sec}
We consider M/M/$c$ queueing systems with staggered setup. Customers arrive at the system in batch according to a Poisson process with rate $\lambda$. We assume that the batch size distribution is $\beta_i$ ($i \in \bbN = \{1,2,\dots \}$) and its generating function is given by $\beta(z)$. In this system, an idle server is turned off immediately. If there are some waiting customers, OFF servers are turned on. Furthermore, a server needs some setup time to be active so as to serve a waiting customer. We assume that the setup time of the OFF servers follows the exponential distribution with mean $1/\alpha_i$ provided that there are $i$ active servers. If a server finishes a job, this server picks a waiting customer if any. If there is not a waiting customer, the server in setup process and idle ones are turned off immediately. It should be noted that in this model a server is in either BUSY or OFF or SETUP. We assume that every customer that enters the system receives service and departs. This means that there is no abandonment. We assume that the service time of jobs follows an exponential distribution with mean $1/\mu$. 

%
%
\begin{figure}
\begin{center}
%
\includegraphics[scale=0.25]{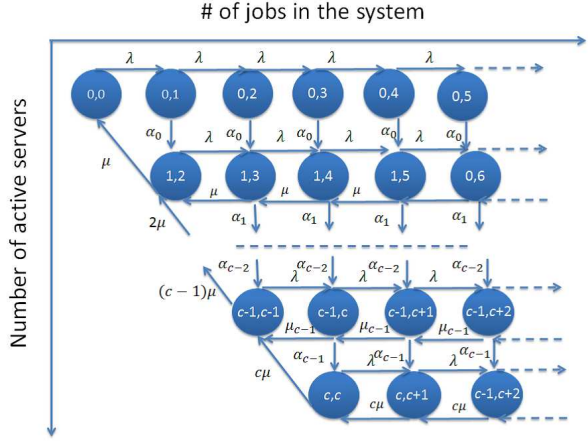}
\end{center}
\caption{State transition diagram ($\beta(z) = z$).}
\label{m:fig}
%
\end{figure}

\section{Analysis of the model}\label{analysis:sec}
\subsection{Generating functions}\label{generating_function:sec}
We present Rouche's theorem which will be repeatedly used in this section. 
\begin{thm}[Rouche's Theorem (see e.g.~\cite{Adan06})]\label{rouche:thm}
Let $D$ denote a bounded region which has a simple closed contour $C$. $f(z)$ and $g(z)$ are two analytic functions on $C$ and $D$. Assume that $|f(z)| < |g(z)|$  on $C$. Then $f(z)$ has in $D$ the same number of zeros as $g(z)$ where all zeros are counted as their multiplicity.
\end{thm}

Let $C(t)$ and $N(t)$ denote the number of busy servers and the number of jobs in the system at time $t$, respectively. Under the assumptions made in Section~\ref{model:sec}, it is easy to see that $\{X(t) = (C(t), N(t)); t \geq 0\}$ forms a Markov chain in the state space 
\[
	\mathcal{S} = \{ (i,j); j \in \bbZ_+, i = 0,1,\dots,\min(c,j) \},
\]
where $\bbZ_+ = \{ 0,1,\dots\}$. See Figure~\ref{m:fig} for the transitions among states for the case of single arrival, i.e., $\beta(z) = z$.

In this paper, we assume that $\rho = \lambda \beta^\prime (1) /(c\mu) < 1$ which is the necessary and sufficient condition for the stability of the Markov chain. 
In what follows, we assume that the Markov chain is ergodic. Under this ergodic condition, let 
\[
\pi_{i,j} = \lim_{t \to \infty} \mathbb{P}(N(t)=i, C(t) = j), \qquad (i,j) \in \mathcal{S},
\]
denote the stationary probability of state $(i,j)$.

The balance equations for states $(0,j)$ ($j \in \bbN$) read as follows. 
\begin{align*}
\lambda \sum_{i=1}^{j} \beta_i \pi_{0,j-i} & = (\lambda + \alpha_0) \pi_{0,j}, \qquad  j \in \bbN.
\end{align*}
Let $\Pi_0 (z) = \sum_{j=0}^\infty \pi_{0,j} z^j$. Multiplying the above equation by $z^j$ and adding over $j \in \bbN$ yields,
\[
\lambda \beta(z) \Pi_0(z) = (\lambda + \alpha_0) (\Pi_0 (z) - \pi_{0,0}),
\]
or equivalently 
\begin{equation}
\label{Pi0z:eq}
	\Pi_0 (z) = \frac{(\lambda + \alpha_0) \pi_{0,0} }{ \lambda + \alpha_0 - \lambda \beta(z)}.
\end{equation}

The balance equation for state $(0,0)$ is given by 
\[
	\lambda \pi_{0,0} = \mu \pi_{1,1}.
\]

This equation is also derived from the balance between flows in and out the group of states $\{ (0,j); j \in \bbZ_+ \}$. 
Indeed, we have 
\[
	\alpha_0 (\Pi_0(1)- \pi_{0,0}) = \mu \pi_{1,1},
\]
leading to 
\[
	\pi_{1,1} = \frac{ \alpha_0 (\Pi_0(1)- \pi_{0,0})  }{\mu} = \frac{\lambda}{\mu} \pi_{0,0}. 
\]

Now, we shift to the case where there is one active server, i.e., $i = 1$. We have 
\begin{equation}
\label{pi_{1,1}:eq}
(\lambda + \mu) \pi_{1,1}  = \alpha_0 \pi_{0,1} + \mu \pi_{1,2} + 2 \mu \pi_{2,2}, \quad j = 1, 
\end{equation}
\begin{equation}
(\lambda + \mu + \alpha_1) \pi_{1,j}  = \lambda \sum_{i=1}^{j-1} \beta_i \pi_{1,j-i} + \alpha_0 \pi_{0,j} + \mu \pi_{1,j+1}, \quad j \geq 2. \label{pi_{1,j}:eq}
\end{equation}
We define the generating for the states with $i=1$ as follows. 
\[
	\Pi_1(z) = \sum_{j=0}^\infty \pi_{1,j+1} z^j.	
\]
$\Pi_1(z)$ represents the generating function of the number of waiting customers while there is one active server.

Multiplying (\ref{pi_{1,1}:eq}) by $z^0$ and (\ref{pi_{1,j}:eq}) by $z^{j-1}$ and taking the sum over $j \in \bbN$ yields, 

\begin{eqnarray}
\label{Pi1(z):eq}
(\lambda + \mu + \alpha_1) \Pi_1(z) - \alpha_1 \pi_{1,1} = \lambda \beta(z) \Pi_1(z)  + \frac{\alpha_0}{z} (\Pi_0(z) - \pi_{0,0}) + \frac{\mu}{z} (\Pi_1 (z) - \pi_{1,1} ) + 2\mu \pi_{2,2}. 
\end{eqnarray}
Arranging (\ref{Pi1(z):eq}) we obtain 
\begin{equation}\label{functional_eq_Pi1z:eq}
	f_1(z) \Pi_1(z) = \alpha \Pi_0 (z) + \alpha_1 z \pi_{1,1} - \alpha_0 \pi_{0,0} - \mu \pi_{1,1} + 2 \mu z \pi_{2,2},
\end{equation}
where $f_1 ( z) = (\lambda +  \mu + \alpha_1) z - \lambda z \beta (z) - \mu$. Because $f_1(0) = -\mu < 0$ and $f_1(1) = \alpha_1 > 0$, $0 < \exists z_1 < 1$ such that $f_1(z_1)=0$. Furthermore, Rouche's theorem 
(Theorem~\ref{rouche:thm}) shows that $z_1$ is the unique root in the unit circle. Indeed, letting $g(z) = (\lambda +  \mu + \alpha_1) z $ and $f(z) = \lambda z \beta(z) + \mu$, $C = \{ z \in \mathbb{C} \ | \ |z| =1 \
$ and $D = \{  z \in \mathbb{C} \ | \ |z| <1 \}$, we see that 
\[ 
|f(z)| \leq \lambda |z| | \beta(z) | + \mu \leq \lambda + \mu < \lambda + \mu + \alpha_1 = |g(z)|, \qquad z \in C.
\]
Thus, applying Rouche's theorem, we have that $f(z)-g(z)$ and $g(z)$ have the same number of zeros.

Since $\Pi_1(z)$ converges in $|z| \leq 1$, letting $z = z_1$ yields,
\begin{equation}
\pi_{2,2} = \frac{ (\mu - \alpha z_1) \pi_{1,1} + \alpha (\pi_{0,0} - \Pi_0 (z_1))  }{2 \mu z_1}. 
\end{equation}
It should be noted that for the case $\beta(z) = z$, i.e., single arrival, we have

\[
	z_1 = \frac{\lambda + \mu + \alpha_1   - \sqrt{(\lambda + \mu + \alpha_1)^2 - 4 \lambda \mu}   }{2 \lambda}.
\]
\begin{remark}
At this point, we have expressed $\Pi_1(z)$ and $\pi_{2,2}$ in terms of $\pi_{0,0}$.
\end{remark}

Furthermore, letting $f_1(z) = (z-z_1) g_1 (z)$, we have that $g_1 (z)$ is an analytic function on the unit circle $|z| < 1$. Substituting this into (\ref{functional_eq_Pi1z:eq}) and arranging the result, we obtain 
\[
	\Pi_1 (z) = \frac{\mu \pi_{1,1} + \alpha_0 \pi_{0,0} + \alpha_1 \widehat{\pi}_0 (z) }{z_1 g_1 (z)}.
\]
where 
\[
	\widehat{\pi}_0 (z) = \frac{\Pi_0 (z) - \Pi_0 (z_1)}{z-z_1}.
\]

Next, we shift to the case where there are $i$ ($2 \leq i \leq c-1$) active servers. 
The balance equations read as follows. 
\begin{equation}
\label{pi_{i,i}:eq}
 (\lambda + i\mu ) \pi_{i,i}  = \alpha_{i-1} \pi_{i-1,i} + i \mu \pi_{i,i+1} + (i+1) \mu \pi_{i+1,i+1},
\end{equation}
\begin{equation}
\label{pi_{i,j}:eq}
 (\lambda + i\mu + \alpha_i) \pi_{i,j}  = \lambda \sum_{k=1}^{j-i} \beta_k \pi_{i,j-k} + i\mu \pi_{i,j+1} + \alpha_{i-1} \pi_{i-1,j},
\end{equation}
for  $j \geq i+1$.
We define the partial generating function for the case of having $i$ active servers as follows. 
\[
	\Pi_i (z) = \sum_{j=i}^\infty \pi_{i,j} z^{j-i}, \qquad i = 2,3,\dots,c-1. 
\]
Multiplying (\ref{pi_{i,i}:eq}) by $z^{0}$ and (\ref{pi_{i,j}:eq}) by $z^{j-i}$ and adding over $j = i,i+1,\dots$, we obtain 

\begin{eqnarray*}
 \lefteqn{ (\lambda + i \mu + \alpha_i) \Pi_i (z) - \alpha_i \pi_{i,i}  =   \lambda  \beta(z) \Pi_i(z) + \frac{i\mu}{z} (\Pi_i (z) - \pi_{i,i} )} \\ 
                                                                         &    &  \mbox{} + \frac{\alpha_{i-1}}{z} (\Pi_{i-1} (z) - \pi_{i-1,i-1}) + (i+1) \mu \pi_{i+1,i+1}, \qquad
\end{eqnarray*}
or equivalently 
\begin{eqnarray}
\label{Pi_i (z):eq}
f_i(z) \Pi_i (z) -\alpha_i z \pi_{i,i}  = (i+1) \mu z \pi_{i+1,i+1} - i\mu \pi_{i,i}  + \alpha_{i-1} (\Pi_{i-1} (z) - \pi_{i-1,i-1}),  \nonumber \\
\end{eqnarray}
where $f_i (z) = (\lambda + i\mu + \alpha_i) z - \lambda z \beta (z) - i\mu $. Since $f_i (0) = -i \mu < 0$ and $f_i (1) = \alpha_i > 0$, $0 < \exists z_i < 1$ such that $f_i (z_i) = 0$. Rouche's theorem also shows that $z_i$ is the unique root inside the unit circle. For the case of single arrival, i.e., $\beta(z) = z$, we have
\[
	z_i = \frac{ \lambda + i \mu + \alpha_i  - \sqrt{ (\lambda + i \mu + \alpha_i)^2 - 4 i \lambda \mu }    }{2 \lambda}.
\]
Putting $z= z_i$ into (\ref{Pi_i (z):eq}), we obtain 
\begin{eqnarray}
\label{pi_i+1,i+1}
 \pi_{i+1,i+1} & = & \frac{ (i\mu - \alpha_i z_i) \pi_{i,i} + \alpha_{i-1} (\pi_{i-1,i-1} - \Pi_{i-1} (z_i) )         }{ (i+1) \mu z_i}, \nonumber \\
                 &    &  i = 1,2,\dots,c-1. 
\end{eqnarray}
\begin{remark}
At this point, we have expressed the generating functions $\Pi_i (z)$ ($i=0,1,\dots,c-1$) and boundary probabilities $\pi_{i,i}$ ($i=0,1,\dots,c$) in terms of 
$\pi_{0,0}$.
\end{remark}

Similar to the case $i=1$, we also have 
\[
	\Pi_i (z) = \frac{\mu \pi_{i,i} + \alpha_0 \pi_{i-1,i-1} + \alpha_i \widehat{\pi}_{i-1} (z) }{z_i g_i (z)}.
\]
where 
\[
	\widehat{\pi}_{i-1} (z) = \frac{\Pi_{i-1} (z) - \Pi_{i-1} (z_i)}{z-z_i}, \qquad g_i(z) = \frac{f_i(z)}{z-z_i}.
\]

Finally, we consider the case $i=c$, i.e., all servers are active. Balance equations are given as follows. 
\begin{align}
\label{pi_{c,c}:eq}
(\lambda + c\mu) \pi_{c,c} & = \alpha_{c-1} \pi_{c-1,c} + c\mu \pi_{c,c+1}, \\
\label{pi_{c,j}:eq}
(\lambda + c\mu ) \pi_{c,j} & = \alpha_{c-1} \pi_{c-1,j} + \lambda \sum_{i=1}^{j-c} \beta_i \pi_{c,j-i} + c\mu \pi_{c,j+1},  \qquad j \geq c+1. 
\end{align}
We define the generating function for the case $i=c$ as follows. 
\[
	\Pi_c (z) = \sum_{j=c}^\infty \pi_{c,j} z^{j-c}.
\]
Multiplying (\ref{pi_{c,c}:eq}) by $z^0$ and (\ref{pi_{c,j}:eq}) by $z^{j-c}$ and summing over $j \geq c$, we obtain 
\begin{eqnarray*}
(\lambda + c\mu ) \Pi_c(z) = \frac{\alpha_{c-1}}{z}   ( \Pi_{c-1}(z) - \pi_{c-1,c-1} )   +  \frac{c \mu}{z} (\Pi_c(z) - \pi_{c,c}) + \lambda \beta(z) \Pi_c(z),
\end{eqnarray*}
or equivalently 
\begin{align*}
 \Pi_c(z) & = \frac{ \alpha_{c-1} ( \Pi_{c-1} (z) - \pi_{c-1,c-1})    -c\mu \pi_{c,c}     }{ f_c (z) }, \\
            & = \frac{ \alpha_{c-1} ( \Pi_{c-1} (z) - \Pi_{c-1} (1)) }{ f_c (z) }
\end{align*}
where $ f_c (z) = (\lambda + c\mu) z - \lambda z \beta(z) - c\mu $ and the second equality is due to the balance between the flows in and out the group of states 
$\{ (c,j); j =c,c+1,\dots \}$. Thus, applying L'Hopital's rule and arranging the results yields
\begin{equation}\label{Pi_c(1):eq}
	\Pi_c(1) = \frac{ \alpha_{c-1} \Pi_{c-1}^\prime (1)   }{c\mu - \lambda \beta^\prime (1)} .
\end{equation}

\begin{remark}
It should be noted that we have expressed $\Pi_i (z)$ ($i=0,1,\dots,c$) in terms of $\pi_{0,0}$, which is uniquely determined by the following normalization condition:
\begin{equation}\label{normalization:cond}
	\sum_{i=0}^c \Pi_i (1) = 1.
\end{equation}
According to (\ref{Pi_c(1):eq}), in order to calculate $\Pi_c (1)$, we need $\Pi_{c-1}^\prime (1)$ which is recursively obtained by Theorem~\ref{factorial:0c-1:thm}.
\end{remark}

\begin{remark}
Once $\pi_{i,i}$ ($i = 0,1,\dots,c$) is determined, we can calculate all the steady state probabilities $\pi_{i,j}$ by a recursive manner via the balance equations. In particular, 
the calculation order is $\{ \pi_{0,j}; j \geq 0 \} \rightarrow \{ \pi_{1,j}; j \geq 1 \} \rightarrow \dots \rightarrow \{\pi_{c,j}; j \geq c\}$.
\end{remark}

In Section~\ref{factorial_moment:sec}, we show some simple recursive formulae for the partial factorial moments.
%
%
\subsection{Factorial moments}\label{factorial_moment:sec}
In this section, we derive simple recursive formulae for factorial moments. Because the generating function $\Pi_0 (z)$ is given in a simple form, its derivatives 
at $z=1$ are also explicitly obtained in a simple form. 

\begin{thm}\label{factorial:0c-1:thm}
The first partial moments of the queue length is recursively calculated as follows.
\begin{eqnarray}\label{first:moment}
\Pi_i^\prime (1) & = & \frac{\alpha_{i-1}}{\alpha_i} \Pi_{i-1}^\prime (1) + \frac{\lambda \beta^\prime (1) - \alpha_i - i\mu}{\alpha_i} \Pi_i(1)  + \frac{ (i+1) \mu \pi_{i+1,i+1} + \alpha_i \pi_{i,i}    }{\alpha_i}, \\
                     &&   \quad i = 1,2,\dots,c-1. \nonumber
\end{eqnarray}
where $\Pi_0^\prime(1) = \pi_{0,0} \lambda \beta^\prime (1)  (\lambda + \alpha_0)/\alpha_0^2$. 
Furthermore, the $n$-th ($n \geq 2$) partial factorial moment is given by  
\begin{eqnarray}\label{nth_moments:eq}
\Pi_i^{(n)} (1) &= & \frac{\alpha_{i-1}}{\alpha_i}   \Pi_{i-1}^{(n)} (1)  +  \frac{n(\lambda \beta^\prime (1) - i\mu - \alpha_i) \Pi_{i}^{(n-1)} (1) }{\alpha_i}  \nonumber \\
                  && \mbox{} + \frac{ \sum_{k=2}^n {}_n C_k \left(\lambda \beta^{(k)} (1) + k \lambda \beta^{(k-1)}(1) \right) \Pi_i^{(n-k)}   }{\alpha_i}, \qquad  i = 1,2,\dots,c-1, \nonumber
\end{eqnarray}
where $\Pi_0^{(n)} (1) =  n! \pi_{0,0} (\lambda \beta^\prime(1))^n (\lambda + \alpha_0)/ \alpha_0^{n+1}$.
\end{thm}
\begin{proof}
%
Differentiating (\ref{Pi_i (z):eq}), we obtain 
\begin{eqnarray*}
	f_i(z) \Pi_i^\prime (z) &  = &   \mbox{} - (i \mu + \alpha_i - \lambda z \beta^\prime (z) ) \Pi_i (z) +   \alpha_{i-1} \Pi_{i-1}^\prime (z)  + \alpha_i \pi_{i,i}  + (i+1)\mu \pi_{i+1,i+1}.
\end{eqnarray*}
Substituting $z=1$ into the above equation and arranging the result yields (\ref{first:moment}).
%
Differentiating (\ref{Pi_i (z):eq}) for $n \geq 2$ times at $z = 1$ and arranging the result, we obtain (\ref{nth_moments:eq}).
\end{proof}

\begin{thm}\label{factorial:c:thm}
We have 
\begin{equation}\label{Pi_c^{(n)} (1):eq}
\Pi_c^{(n)} (1) = \frac{   A_n   }{ (n+1) (c\mu - \lambda \beta^\prime (1) ) }, \quad n \in \bbN,
\end{equation}
where 
\begin{eqnarray*}
A_n = \alpha_{c-1} \Pi_{c-1}^{(n+1)} (1) +  \sum_{k=2}^{n+1} {}_{n+1} C_k \left( \lambda k \beta^{(k-1)} (1) + \lambda \beta^{(k)} (1)   \right) \Pi^{(n+1-k)}_c (1).
\end{eqnarray*}
\end{thm}

\begin{proof}
We have 
\begin{equation*}
 f_c(z)  \Pi_c(z) = \alpha_{c-1} ( \pi_{c-1} (z) - \pi_{c-1,c-1})    -c\mu \pi_{c,c}.
\end{equation*}
Differentiating this equation $n \geq 1$ times, we obtain 

\begin{eqnarray*}
    f_c(z)  \Pi_c^{(n)} (z)  +  \sum_{k=1}^n {}_n C_k f_c^{(k)} (z) \Pi_c^{(n-k)} (z)  =  \alpha_{c-1} \Pi_{c-1}^{(n)} (z),
\end{eqnarray*}
where $\Pi_c^{(-1)} (z) = 0, \forall \ |z| < 1$.  Arranging this equation leads to 
\begin{eqnarray}\label{nthdiff_pic(n)z:eq}
\Pi_c^{(n)} (z)  =  \frac{ \alpha_{c-1} \Pi_{c-1}^{(n)} (z) - \sum_{k=1}^n {}_n C_k f_c^{(k)} (z) \Pi_c^{(n-k)} (z)  }{ f_c(z) }. 
\end{eqnarray}

We observe inductively that both the denominator and numerator in the right hand side of (\ref{nthdiff_pic(n)z:eq}) vanish at $z=1$. 
Thus, applying L'Hopital's rule and arranging the result, we obtain  (\ref{Pi_c^{(n)} (1):eq}).
\end{proof}

\begin{remark}
It should be noted that in order to obtain the $n$-th factorial moment $\Pi_c^{(n)} (1)$, we need to have the $(n+1)$-th factorial moment $\Pi_{c-1}^{(n+1)} (1)$.
Fortunately, $\Pi_{c-1}^{(n+1)} (1)$ is expressed in terms of $\Pi_{0}^{(n+1)} (1)$ which is explicitly obtained for any $n$ according to Theorem~\ref{factorial:0c-1:thm}.
\end{remark}

\begin{remark}
It should be noted that when $\alpha_i = \alpha$ ($i = 0,1,\dots,c-1$), our results reduce to those presented in~\cite{phungduc14b}.
\end{remark}



\section{Waiting time distribution}\label{waiting_time_distribution:sec}
This section is devoted to the waiting time distribution of an arbitrary customer. To this end, we first find the steady state probability $p_{i,n-1}$ that an arriving customer finds $i$ servers in active model and $n-1$ ($n \geq 1$) customers standing before him. We then find the conditional waiting time $W_{i,n}$ of a tagged customer that finds $i$ active servers and $n-1$ customers stand before him. Let $\widetilde{W}_{i,n} (s)$ denote the LST of  $W_{i,n}$. Let $W$ denote the waiting time of an arbitrary customer and $\widetilde{W} (s)$ denote the LST of $W$. 
We then have 
\[
	\widetilde{W} (s) = \sum_{i=0}^c \sum_{n=i+1}^\infty p_{i,n-1} \widetilde{W}_{i,n} (s).
\]
In Artalejo et al.~\cite{Artalejo05}, explicit expression for $\widetilde{W}_{i,n} (s)$ is obtained. 
In fact, $\widetilde{W}_{i,n} (s)$ is the first passage time from state $(i,n)$ to the boundary $(i,i), i = 0,1,\dots,c-1,c$. 
Thus, we can obtained the waiting time distribution by inverting the LST.   
\subsection{Computation of $p_{i,n}$}
Because $p_{i,n}$ denotes the probability that an arriving customer finds that there are $i$ active servers and himself at the order of $n$ (to depart from the system). We have 
\[
	p_{i,n} = \sum_{j=1}^{n} \pi_{i,n-j} r_j, 
\]
where $r_j$ is the probability that an arriving customer finds himself at the $j$-th in the batch. According to Burke~\cite{Burke} and Cromie et al.~\cite{Cromie}, we have 
\[
	r_j = \frac{1}{{\rm E}[B]} \sum_{i=j}^\infty \beta_i, \qquad j = 1,2,\dots,
\]
where ${\rm E}[B] = \beta^\prime (1)$ is the mean batch size. 
\subsection{Algorithm for the stationary distribution}
In this section, we present an algorithm calculating all the joint steady state probabilities. Since $\pi_{i,i}$ ($i = 0,1,\dots,c$) are obtained. We can calculate all other steady state probability using a recursive algorithm. Indeed, $\pi_{0,n}$ is recursively obtained if $\pi_{0,0}$ is given. Given that $\pi_{0,n}$ is known for any $n$ and that $\pi_{1,1}$ is known, we can recursively obtain all the probabilities $\pi_{1,n}$ for $n \geq 1$. Similarly, we could obtain all the probabilities $\pi_{i,n}$ $(i,n) \in \mathcal{S}$.

\section{Conditional Decomposition}\label{decomposition}
We have derived the following result.
\begin{eqnarray*}
 \Pi_c(z) & = & \frac{ \alpha_{c-1} ( \pi_{c-1} (z) - \pi_{c-1,c-1})    -c\mu \pi_{c,c}     }{ f_c(z) }, \\ 
 \Pi_c(1) & = & \frac{ \alpha_{c-1} \Pi_{c-1}^\prime (1)   }{c\mu - \lambda}.
\end{eqnarray*}

Let $Q^{(c)}$ denote the conditional queue length given that all $c$ servers are busy, i.e.,
\[
	\mathbb{P} (Q^{(c)} = i) = \mathbb{P} (N = i + c \ | \ C = c),
\]
where $N$ and $C$ are the number of customers in the system and that of busy servers in the steady state, respectively. 
Let $P_c(z)$ denote the generating function of  $Q^{(c)}$. It is easy to see that 
\begin{align*}
	P_c(z) & = \frac{\Pi_c(z)}{\Pi_c(1)} \\
            & = \frac{  \alpha_{c-1} ( \Pi_{c-1} (z) - \pi_{c-1,c-1})    -c\mu \pi_{c,c}     }{\alpha_{c-1} \Pi_{c-1}^\prime(1) (z-1)} g(z) \\
	        & = \frac{\Pi_{c-1} (z) - \Pi_{c-1}(1)}{\Pi_{c-1}^\prime(1)(z-1)}   g(z) \\
	        & = \frac{\sum_{j=1}^\infty \pi_{c-1,c-1+j} (z^j - 1)  }{\Pi_{c-1}^\prime(1)(z-1)} g(z) \\
	        & = \frac{\sum_{j=1}^\infty \pi_{c-1,c-1+j} \sum_{i=0}^{j-1} z^i  }{\Pi_{c-1}^\prime(1)} g(z)\\
	        & = \frac{ \sum_{i=0}^\infty \left(\sum_{j = i+1}^\infty \pi_{c-1,c-1+j} \right)  z^i }{\Pi_{c-1}^\prime(1)} g(z),
\end{align*}
where we have used $c\mu \pi_{c,c} = \alpha_{c-1} (\Pi_{c-1} (1) - \pi_{c-1,c-1}) $ in the second equality and 
\[
g(z) = \frac{(c\mu - \lambda \beta^\prime(1))(z-1)}{ (c\mu + \lambda )z - \lambda z \beta(z) - c\mu}.
\]
It should be noted that $g(z)$ is the generating function of the number of waiting customers in the conventional M$^{X}$/M/$c$ system without setup time (denoted by $Q^{(c)}_{ON-IDLE}$) under the condition that $c$ servers are busy.

We give a clear interpretation for the generating function:
\[
	\frac{ \sum_{i=0}^\infty \left(\sum_{j = i+1}^\infty \pi_{c-1,c-1+j} \right)  z^i }{\Pi_{c-1}^\prime(1)}.
\]
For simplicity, we define  
\begin{eqnarray*}
	q_{c-1,i} =   \frac{\sum_{j = i+1}^\infty \pi_{c-1,c-1+j}}{\Pi_{c-1}^\prime(1)}, \qquad i \in \bbZ_+.
\end{eqnarray*}

We have
\[
	\sum_{j=i+1}^\infty \pi_{c-1,c-1 + j} = \mathbb{P} ( N - C > i \ | \ C= c-1) \mathbb{P} (C=c-1).
\]
Thus, we have 
\[
	q_{c-1,i} = \frac{\mathbb{P} ( N - C > i \ | \ C= c-1)}{ \mathbb{E} [N - C \ | \ C = c-1] }.
\]

It should be noted that $N-C$ is the number of waiting customers. Thus, the discrete random variable with the distribution $q_{c-1,i}$ ($i=0,1,2,\dots$) is the 
the probability that a waiting waiting customer find $i$ other customers waiting in front of him under the condition that there are $c-1$ active servers (See~Burke~\cite{Burke}). Let $Q_{Res}$ denote this random variable. 

Thus our decomposition result is summarized as follows. 
\[
	Q^{(c)} \,{\buildrel d \over =}\  Q^{(c)}_{ON-IDLE}  +    Q_{Res}. 
\]
\begin{remark}
Tian et al.~\cite{Tian99,Tian03b} obtain a similar result for a multiserver model with Poisson arrival and vacation, i.e., $\alpha_i = (c-i) \alpha$ and $\beta(z) = z$. 
However, the random variable with the distribution $p_{c-1,i}$ here is not given a clear physical meaning in~\cite{Tian99,Tian03b}.
\end{remark}


\section{Special Cases and Variant models} \label{variant models}

\subsection{Staggered setup model}
In staggered setup policy, only one server can be allowed to be in setup process at a time. Thus, this model is a special case of the model in this paper where $\alpha_i = \alpha$ ($i = 0,1,\dots,c-1$)~\cite{phungduc14c}.

%
Some simpler results could be obtained if we restrict ourself to the case of single arrival, i.e., $\beta(z) = z$. 
In this section is devoted to the decomposition property of the queue length where we show the single server system in Section~\ref{single_server:sec} and discuss the multiserver model in Section~\ref{multiserver:decompose:sec}. 
\subsubsection{Single server}\label{single_server:sec}
We consider the single server case. The partial generating functions are given as follows.
\[
	\Pi_0 (z) = \frac{(1- \rho) \alpha }{\lambda + \alpha - \lambda z}, \qquad \Pi_1 (z) = \frac{(1-\rho) \lambda \alpha}{ (\mu- \lambda z) (\lambda + \alpha - \lambda z)  },   
\]
where $\rho = \lambda/\mu$.
Let $\Pi(z)$ denote the generating function of the number of waiting customers. We have 
\[
	\Pi (z) = \Pi_0 (z) + \Pi_1(z) = (1-\rho) \left( 1 + \frac{\rho}{1-\rho z} \right) \frac{\alpha}{\lambda + \alpha - \lambda z}.
\]
It should be noted that 
\[
	(1-\rho) \left( 1 + \frac{\rho}{1-\rho z} \right)
\]
and 
\[
	\frac{\alpha}{\lambda + \alpha - \lambda z}
\]
represent the generating function of the number of waiting customers in the corresponding M/M/1 queue without setup time and that of customers arriving in the remaining setup time, respectively. 
Thus, we have 
\[
	L \,{\buildrel d \over =}\ L_1 + L_2,
\]
where the $L$ is the queue length of the current model while $L_1$ and $L_2$ represent the queue length of the conventional M/M/1 queue and the number of customers that arrive during the remaining setup time.
\subsubsection{Multiserver}\label{multiserver:decompose:sec}
In this section, we investigate the decomposability of the queue length. In particular we answer the question: does equation (\ref{multiple_serv:eq}) hold?
\begin{equation}\label{multiple_serv:eq}
	L \,{\buildrel d \over =}\ L_1 + L_2,
\end{equation}
where $L_1$ is the queue length of the M/M/$c$ without setup time and $L_2$ is the number of customers that arrive to the queue during the remaining setup time. 

The generating function for the number of waiting customers in the conventional M/M/$c$ queueing system is given by
$1-C(c \rho,c) + C(c \rho,c)(1-\rho)/(1-\rho z)$ where $\rho = \lambda/(c\mu)$ and $C(c \rho,c)$ is the Erlang C formula for the waiting probability in the conventional M/M/$c$ system without setup time. Therefore, if the decomposition result is established the generating function of the number of waiting customers in the system with setup time $\Pi (z)$ must be given by the following formula.
\begin{equation}
\label{mmc_decompose:eq}
	\Pi (z) = \frac{\alpha}{\alpha + \lambda - \lambda z} \left( 1- C(c \rho,c) + C(c \rho,c) \frac{1- \rho}{1- \rho z}  \right). 
\end{equation}

In~\cite{Gandhi10} the authors state that the decomposition property is held for the model meaning that (\ref{mmc_decompose:eq}) is true.

We prove this property. Indeed for the case where $\beta(z) = z$ after some tedious calculations, we find that 
\[
	\Pi_i (z) = \pi_{i,i} \frac{\lambda + \alpha}{\lambda + \alpha - \lambda z}, \qquad \pi_{i,i}=  \pi_{0,0} \left( \frac{\lambda}{\mu} \right)^i \frac{1}{i!}, 
\]
for $i = 0,1,\dots,c-1$ and
\[
	\pi_{c,c}=  \pi_{0,0} \left( \frac{\lambda}{\mu} \right)^c \frac{1}{c!}, \quad \Pi_c (z) = \pi_{c,c} \frac{\lambda + \alpha}{(1-\rho z)(\lambda + \alpha - \lambda z)}.      
\]
It follows from $\Pi (z) = \sum_{i=0}^c \Pi_i(z) $ and $\Pi (1) = 1$ that (\ref{mmc_decompose:eq}) is true.

From the decomposition result for the queue length, we obtain the decomposition result for the waiting time via distributional Little's law. 
In particular, we have 
\begin{equation}\label{multiple_serv_waiting:eq}
	W \,{\buildrel d \over =}\ W_1 + W_2,
\end{equation}
where $W$ denotes the waiting time in the current system while $W_1$ and $W_2$ are the waiting time in the corresponding M/M/$c$ system without setup time and the setup time, respectively.

\begin{remark}
From the generating function, we obtain explicit expressions for the joint stationary distribution as follows.
\[
	\pi_{i,j} = \pi_{i,i} \left( \frac{\lambda}{\lambda+\alpha} \right)^{j-i}, \qquad j = i,i+1,\dots, i = 0,1,\dots,c-1.
\]
Furthermore, if $\rho \neq  \varphi_0 = \lambda/(\lambda + \alpha)$, we have 
\[
	\pi_{c,c+k} =  \pi_{c,c} \left( \frac{\varphi_0^{k+1} - \rho^{k+1} }{\varphi_0 - \rho}   \right), \qquad k \geq 0.
\]
If $\rho = \lambda/(\lambda + \alpha)$, we have 
\[
	\pi_{c,c+k} =  \pi_{c,c} (k+1) \rho^k, \qquad k \geq 0.
\]

\end{remark}

\subsection{Vacation model}
A special case is the one with vacation. The model Poisson arrival is presented in~\cite{Tian99}. In vacation model, a server goes to vacation upon completion of a service and there is not a waiting customer. 
Assuming that the vacation period is exponentially distributed with mean $1/\alpha$. Thus, when there are $i$ active servers and some waiting customers, $c-i$ servers come back to service with rate $(c-i) \alpha$. 
Thus this vacation model is equivalent to our setup model where the setup time is exponentially distributed with mean $1/\alpha_i$ where $\alpha_i = (c-i) \alpha$ provided that there are $i$ active servers.

See Figure~\ref{state_dependent_setup:fig} for the transition among states for the model with state-dependent setup and individual arrivals. 
\begin{figure}
\begin{center}
%
\includegraphics[scale=0.25]{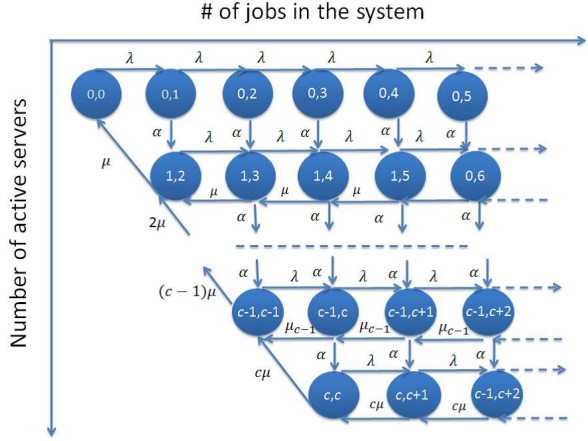}
\end{center}
\caption{State transition diagram.}
\label{state_dependent_setup:fig}
%
\end{figure}

\subsection{Model with queue-length-dependent setup}
Another variant is the model where the number of setting up servers depends on the number of waiting customers \cite{Gandhi10b}. In particular, the setup rate will be $\min(j-i,c-i) \alpha$ provided that there are $i$ active servers and $j$ customers in the system~\cite{phungduc14c}. See Figure~\ref{demand_dependent_setup:fig} for the transition among states for the case of individual arrivals, i.e., $\beta(z) = z$. This model is more complex due to the inhomogeneity of boundary states where the number of customers in the system $j \leq c$. However, we can treat this model by a similar approach with a minor modification. In particular, we may define 
the generating functions for homogeneous part, i.e., $j \geq c$.
\[
	\Pi_i (z) = \sum_{j = c}^\infty \pi_{i,j} z^{j-i}, \qquad i = 0,1,\dots,c-1,c.
\]
This results in a set of equations of generating functions for the states $\{(i,j); j \geq c \}$. In addition, we have balance equations for the states $\{(i,j); i \leq j \leq c \}$

\begin{figure}
\begin{center}
%
\includegraphics[scale=0.5]{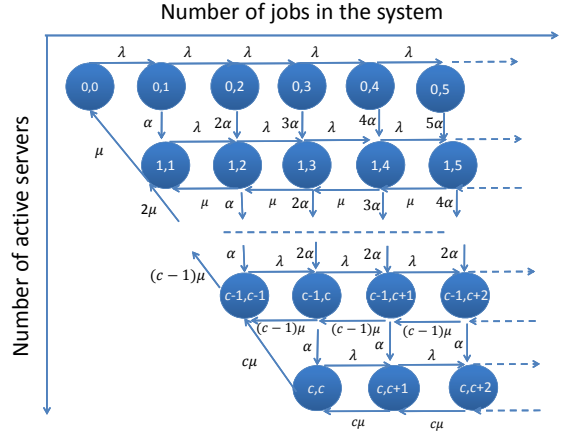}
\end{center}
\caption{State transition diagram.}
\label{demand_dependent_setup:fig}
%
\end{figure}

\section{Performance Measures and Numerical Results}\label{numerical:sec}

\subsection{Power Consumption}
The cost per unit time for each state: SETUP, ON and IDLE of a server is set as follows: $C_{setup} = 1, C_{run} = 1$ and $C_{idle} = 0.6$. The power consumption of our system with staggered setup is given by 
\[
		P_{ON-off} = C_{setup} (1 - \sum_{i=0}^{c-1} \pi_{i,i} - \Pi_c(1) ) + C_{run} c \rho,
\]
where $c \rho = \lambda/\mu$ is the mean number of running servers. We plot four curves corresponding to the cases $\alpha = 0.1,1,10$ and 100. For comparison, we also plot the curves for the conventional M/M/$c$ queue under the same setting. It should be noted that in the conventional M/M/$c$ system, an idle server is not turned off. As a result, the cost for power consumption is given by 
\[
		P_{ON-idle} = C_{run} c \rho + C_{idle} (c- c \rho).
\]

\subsection{Total Cost}
The mean number of waiting customers is given by 
\[
	{\rm E} [Q] = \sum_{i=0}^c \Pi^\prime(1).
\]

We consider a cost function taking into account both power consumption and performance (mean number of waiting customers). Our aim is to to investigate the the characteristics of the cost function. Cost function for the ON-OFF model is given by 
\[
	C_{ON-off} = P_{ON-off} + \frac{1}{\delta} {\rm E}[Q].
\]
On the other hand, the cost function for ON-Idle model is given by 
\[
	C_{ON-idle} = P_{ON-idle} + \frac{1}{\delta} {\rm E}[Q_i],
\]
where ${\rm E}[Q_i]$ is the mean queue length $M^X$/M/c queue without setup time which could be obtained from the analysis in~\cite{Cromie}.

\begin{figure}[htbp]
\begin{tabular}{cc}
\begin{minipage}{0.5\hsize}
\begin{center}
\includegraphics[scale=0.55]{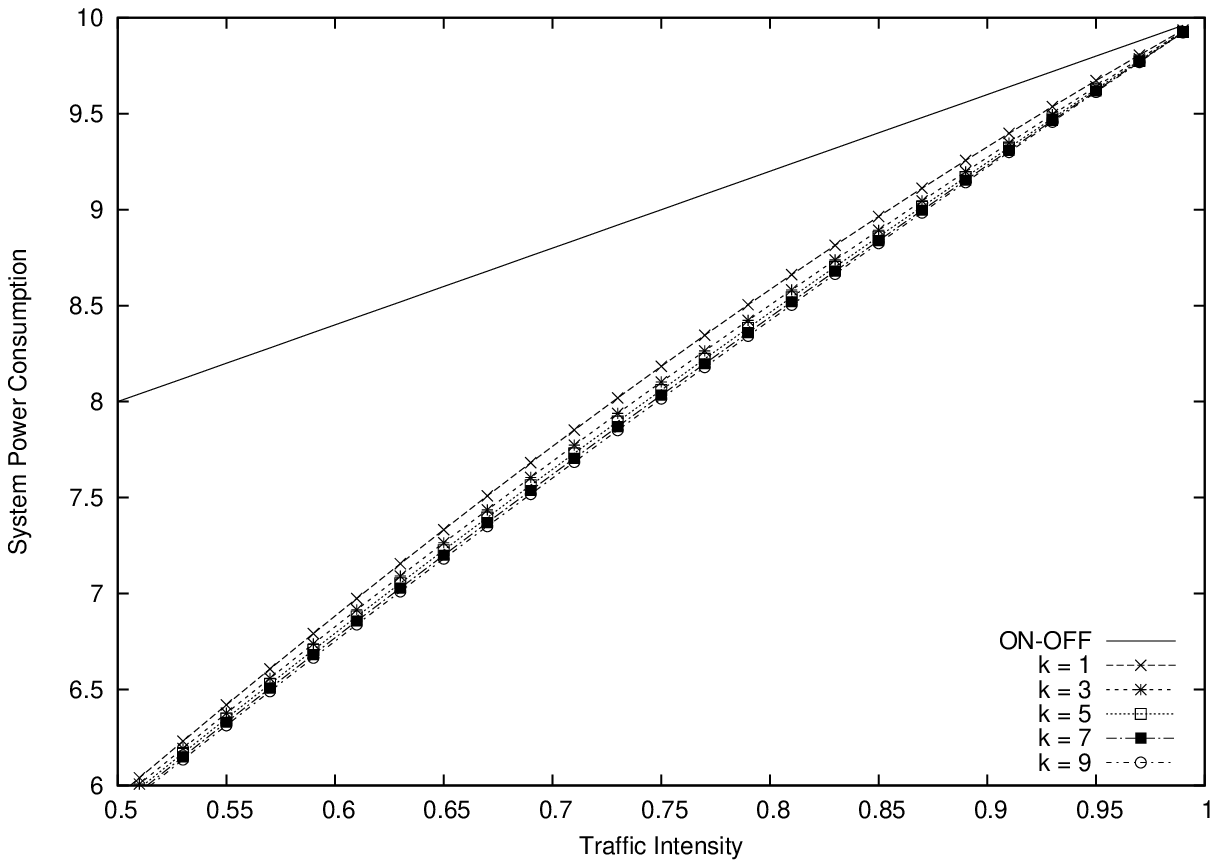}
\caption{Power Consumption $\alpha = 0.1$.}
\label{Power_fixed_batch_rho_change_alpha01_costsetup_equal_costrun:fig}
\end{center}
\end{minipage}
\begin{minipage}{0.5\hsize}
\begin{center}
\includegraphics[scale=0.55]{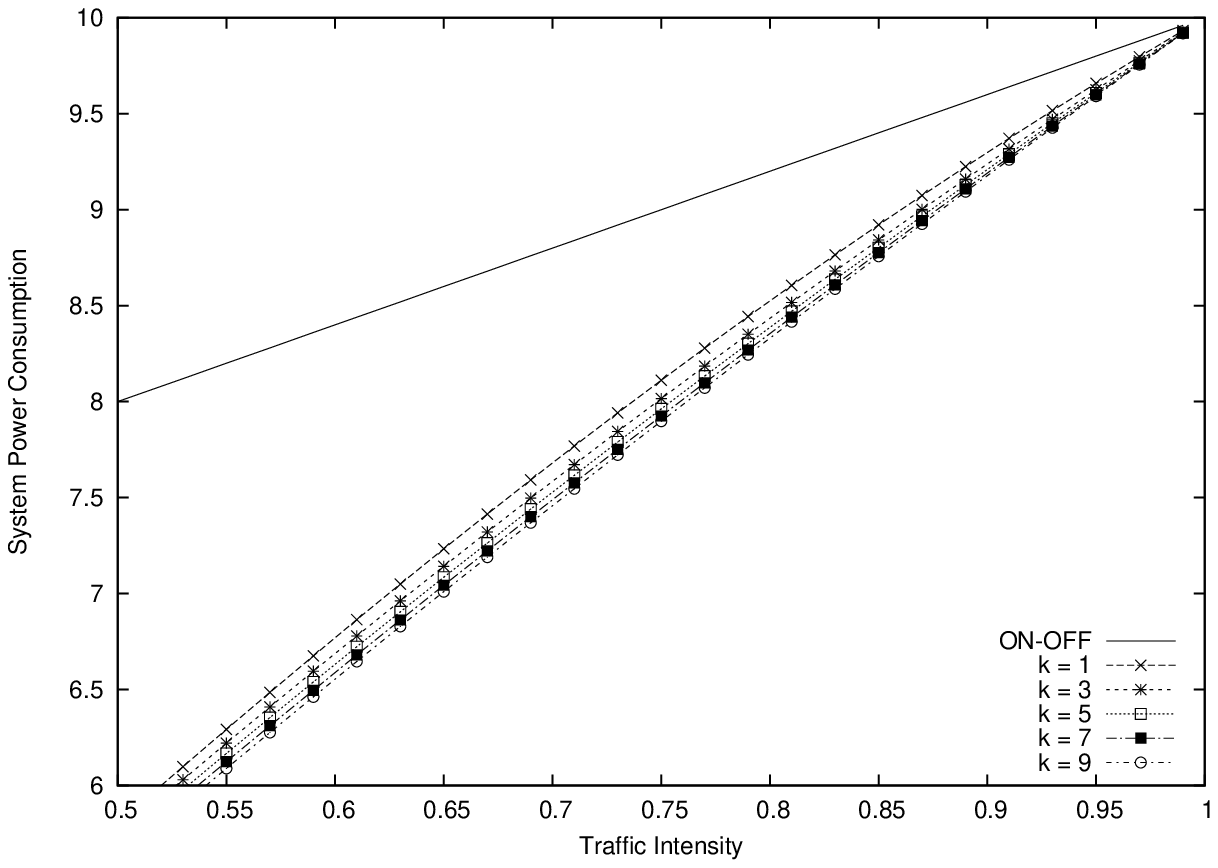}
\caption{Power Consumption $\alpha = 1$.}
\label{Power_fixed_batch_rho_change_alpha1_costsetup_equal_costrun:fig}
\end{center}
\end{minipage}
\end{tabular}
\end{figure}

In this section, we consider the case where $\alpha_i = \alpha$, i.e., staggered setup policy. In all the figures, the curves for On-Idle policy is indicated by "On-Idle" and other curves are of the On-Off model.  
\subsection{Power consumption}
In this section we investigate the power consumption against the traffic intensity. Figures~\ref{Power_fixed_batch_rho_change_alpha01_costsetup_equal_costrun:fig}, \ref{Power_fixed_batch_rho_change_alpha1_costsetup_equal_costrun:fig} and 
\ref{Power_fixed_batch_rho_change_alpha10_costsetup_equal_costrun:fig} against the traffic intensity for $\alpha = 1$. We observe from the three figures that the ON-Off policy always outperform the On-Idle policy. However, from the performance point of view, the waiting time in the former is longer than the latter. We will investigate the impact of setup time on the total cost of the system next section. An important observation is that keeping the traffic intensity the same, power consumption decreases with the batch size. This implies that it is more efficient to design the system so that customers arrive in group with large batch size.

\begin{figure}[htbp]
\begin{tabular}{cc}
\begin{minipage}{0.5\hsize}
\begin{center}
\includegraphics[scale=0.55]{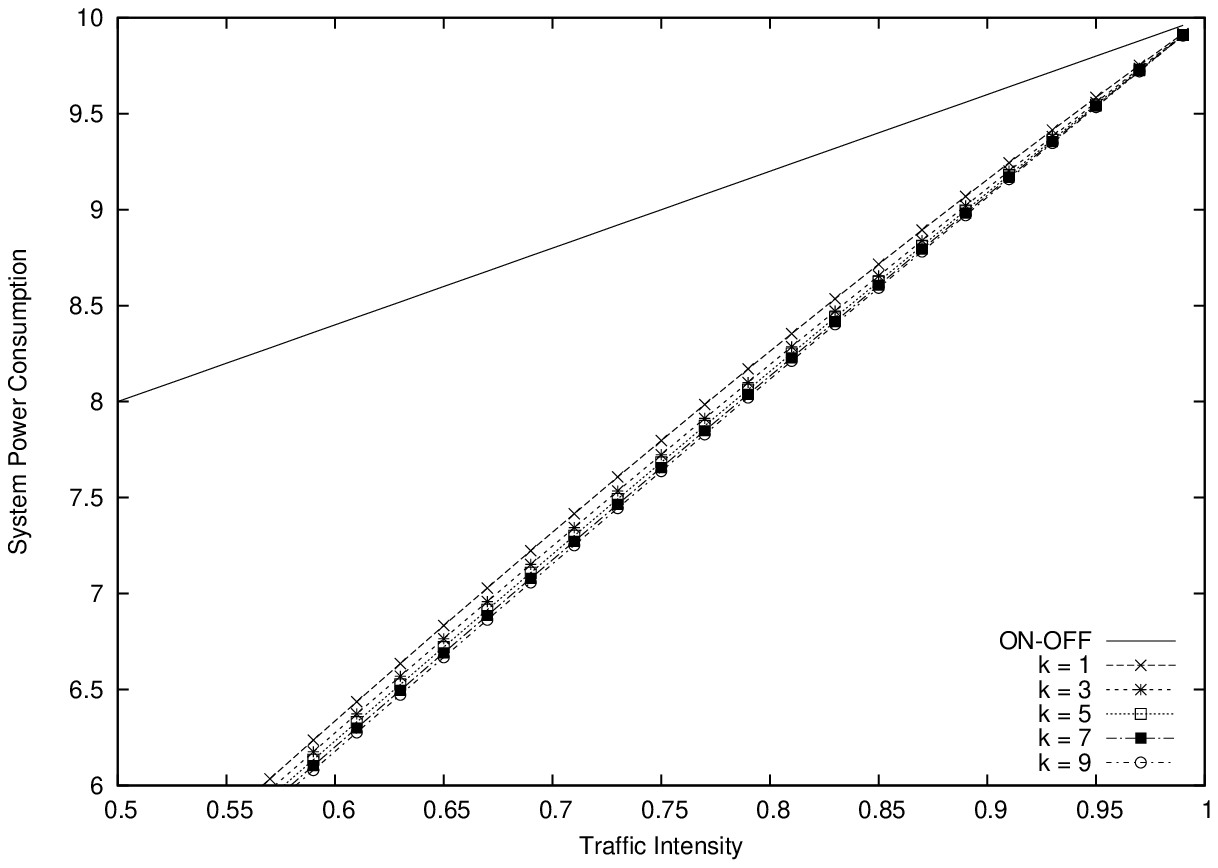}
\caption{Power Consumption $\alpha = 10$.}
\label{Power_fixed_batch_rho_change_alpha10_costsetup_equal_costrun:fig}
\end{center}
\end{minipage}
\begin{minipage}{0.5\hsize}
\begin{center}
\includegraphics[scale=0.55]{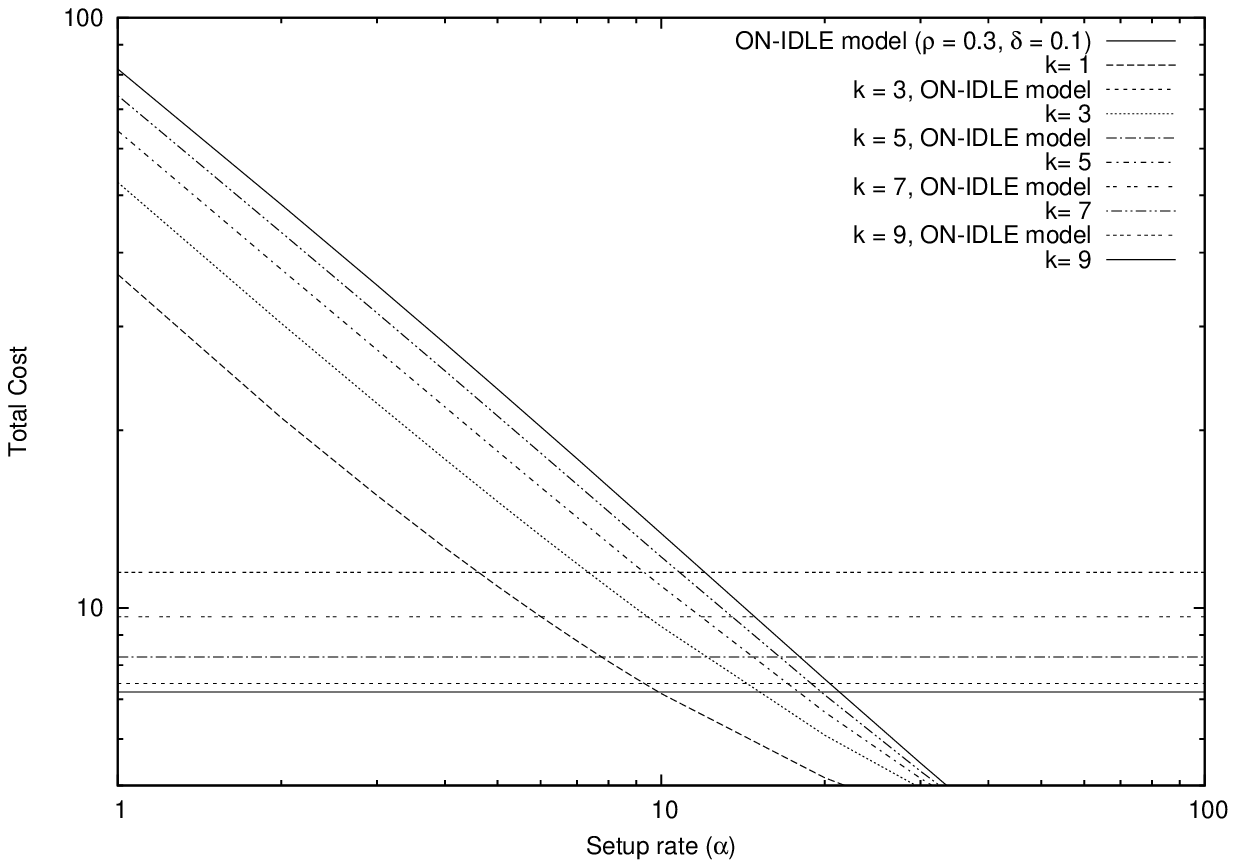}
\caption{Total Cost $\rho = 0.3, \delta = 0.1$.}
\label{power_qos_tradeoff_vs_alpha_rho03_beta01:fig}
\end{center}
\end{minipage}
\end{tabular}
\end{figure}

\subsection{Cost function}

In this section, we investigate the cost function against various parameters. Figures~\ref{power_qos_tradeoff_vs_alpha_rho03_beta01:fig}, \ref{power_qos_tradeoff_vs_alpha_rho03_beta1:fig} and \ref{power_qos_tradeoff_vs_alpha_rho03_beta10:fig} represent the cost function against the setup rate $\alpha$ for $\delta = 0.1$, 1, and 10, respectively, provided that $\rho = 0.3$. In data center, a server typically operates under the load of 40\%~\cite{Schwartz12}. Thus, it may be interesting to investigate the cost function around this value. It should be noted that $\delta = 0.1$ corresponds to the case where the importance of performance, i.e. mean queue length is 10 times bigger than that of the power consumption while $\delta = 10$ represents opposite case where power consumption is given priority. For comparison we also plot the cost function for the ON-Idle model. We observe from the three graphs that there exists some $\alpha_{0.3}$ for which the ON-OFF model is more efficient than the ON-IDLE one when $\alpha > \alpha_{0.3}$. 

 Figures~\ref{power_qos_tradeoff_vs_rho_ap1delta01:fig}, \ref{power_qos_tradeoff_vs_rho_ap1delta1:fig} and \ref{power_qos_tradeoff_vs_rho_ap1delta10:fig} 
show the cost function against the traffic intensity $\rho$ for $\delta = 0.1$, 1 and $10$, respectively. We observe from Figure~\ref{power_qos_tradeoff_vs_rho_ap1delta01:fig} that the ON-Idle model outperforms that ON-Off model. This implies that when the importance is placed on the performance ($\delta = 0.1$), it is better to keep the servers ON all the time. On the other hand, we observe from Figure~\ref{power_qos_tradeoff_vs_rho_ap1delta01:fig}  that the ON-Off model is always better than the ON-Idle one for $\delta = 10$. This implies that when the importance is placed on the power consumption, it is better to adopt the ON-Off model.

\begin{figure}[htbp]
\begin{tabular}{cc}
\begin{minipage}{0.5\hsize}
\begin{center}
\includegraphics[scale=0.55]{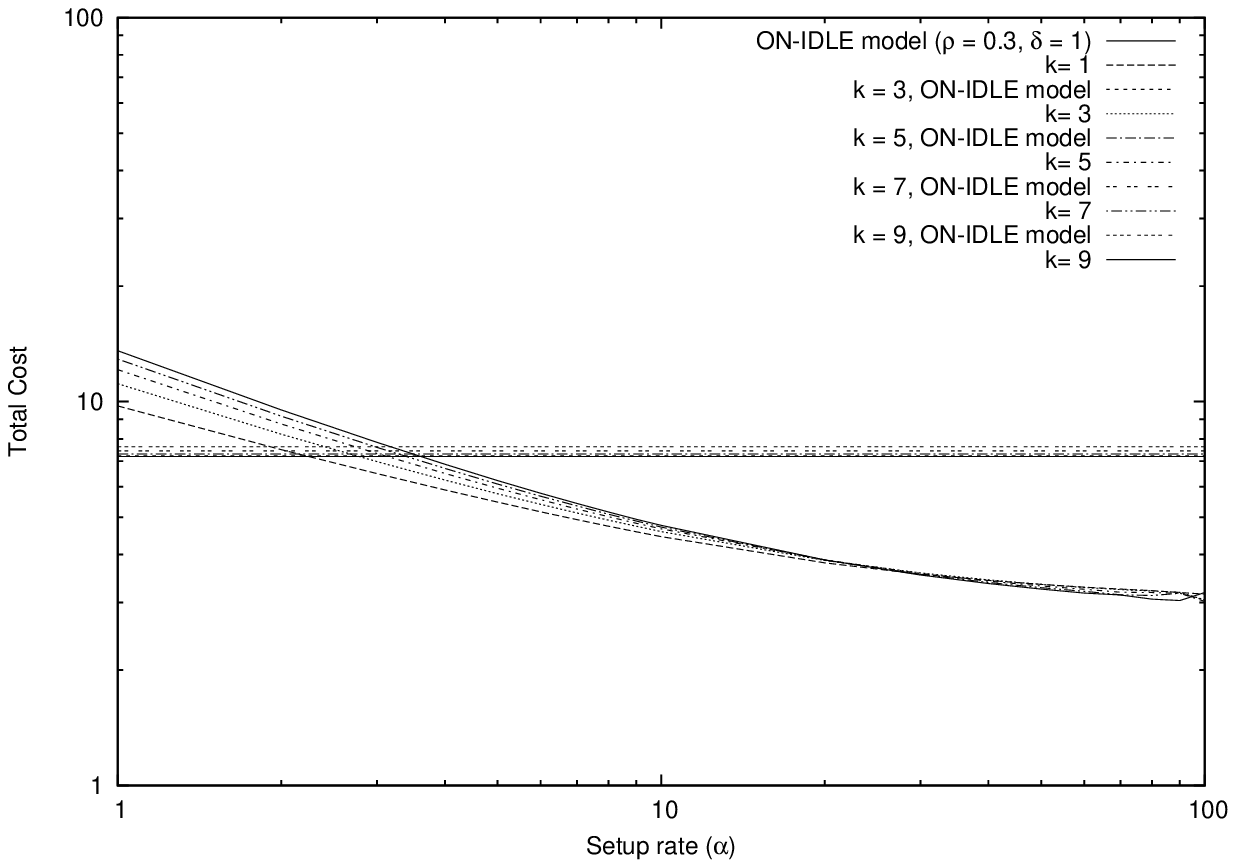}
\caption{Total Cost $\rho = 0.3, \delta = 1$.}
\label{power_qos_tradeoff_vs_alpha_rho03_beta1:fig}
\end{center}
\end{minipage}
\begin{minipage}{0.5\hsize}
\begin{center}
\includegraphics[scale=0.55]{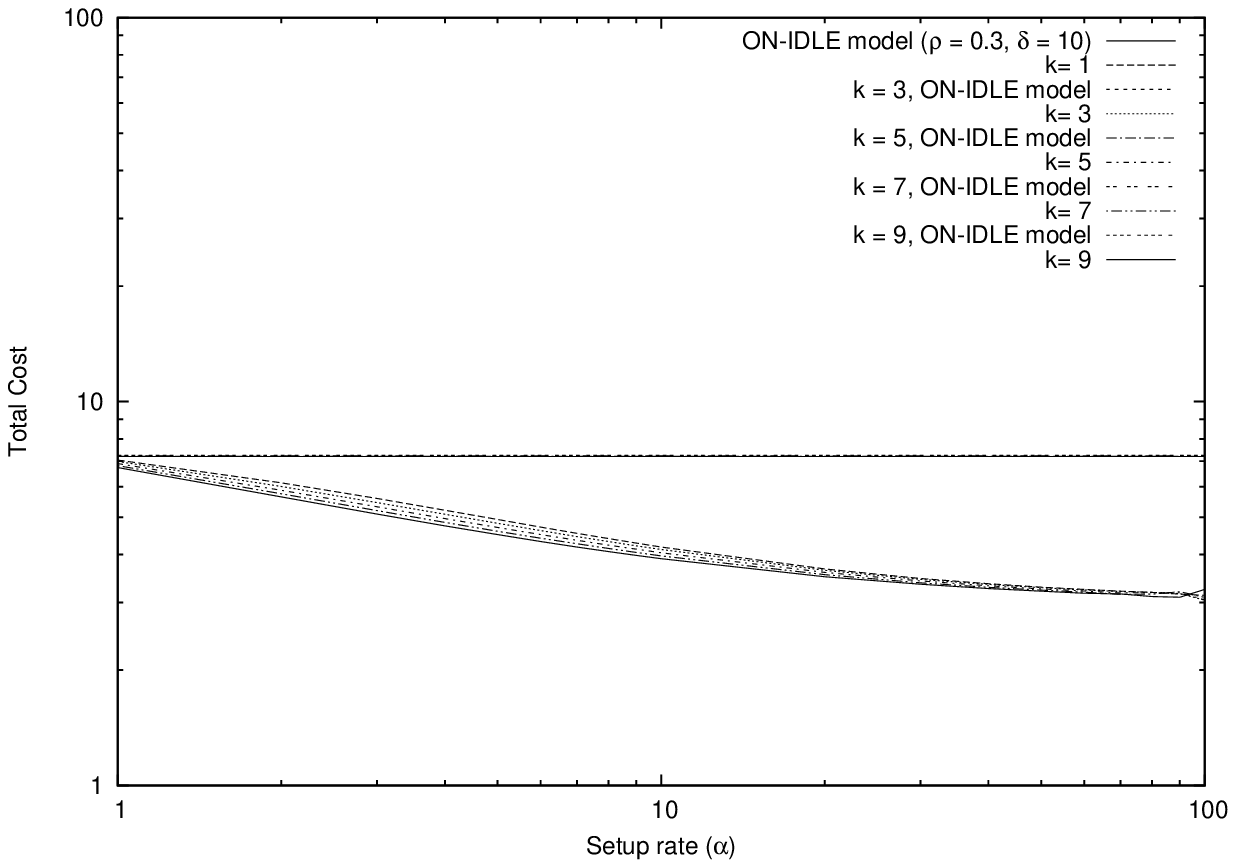}
\caption{Total Cost $\rho = 0.3, \delta = 10$.}
\label{power_qos_tradeoff_vs_alpha_rho03_beta10:fig}
\end{center}
\end{minipage}
\end{tabular}
\end{figure}

\begin{figure}[htbp]
\begin{tabular}{cc}
\begin{minipage}{0.5\hsize}
\begin{center}
\includegraphics[scale=0.55]{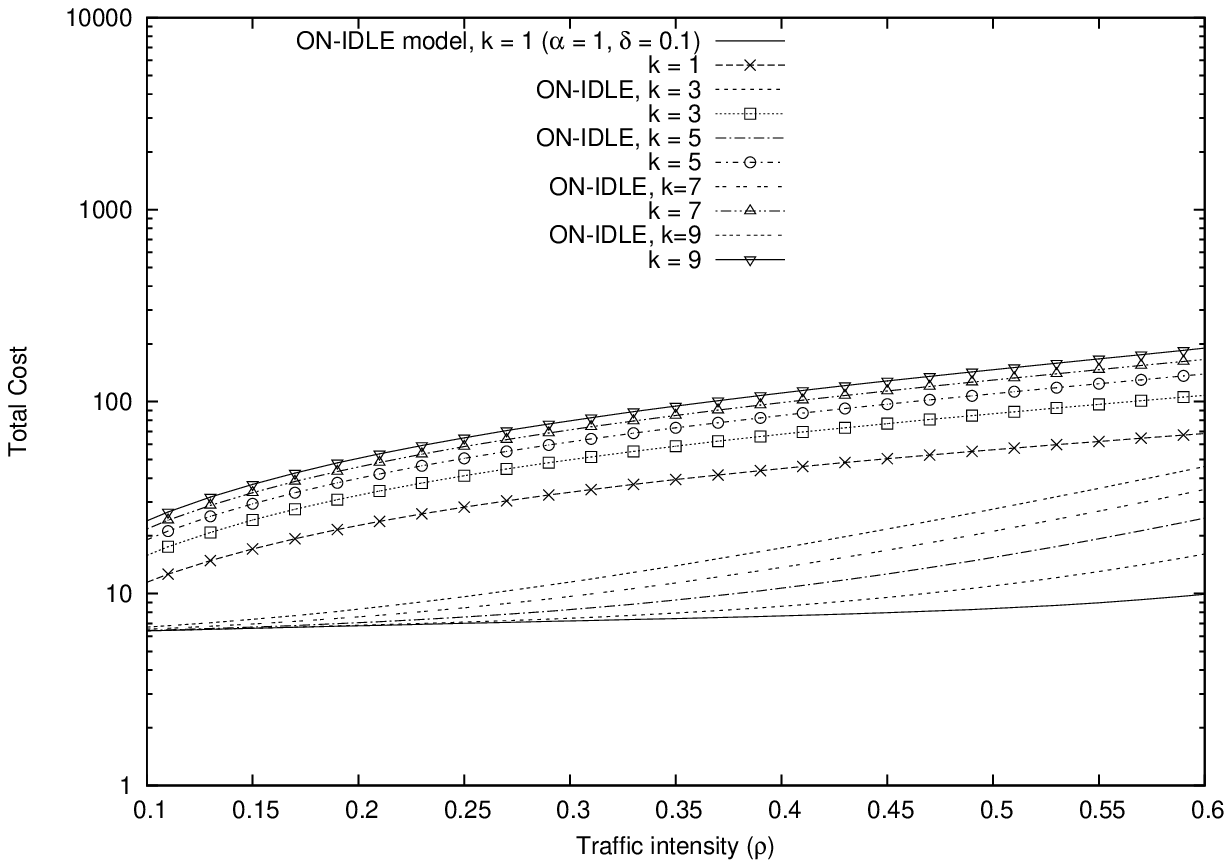}
\caption{Total Cost $\alpha = 1, \delta = 0.1$.}
\label{power_qos_tradeoff_vs_rho_ap1delta01:fig}
\end{center}
\end{minipage}
\begin{minipage}{0.5\hsize}
\begin{center}
\includegraphics[scale=0.55]{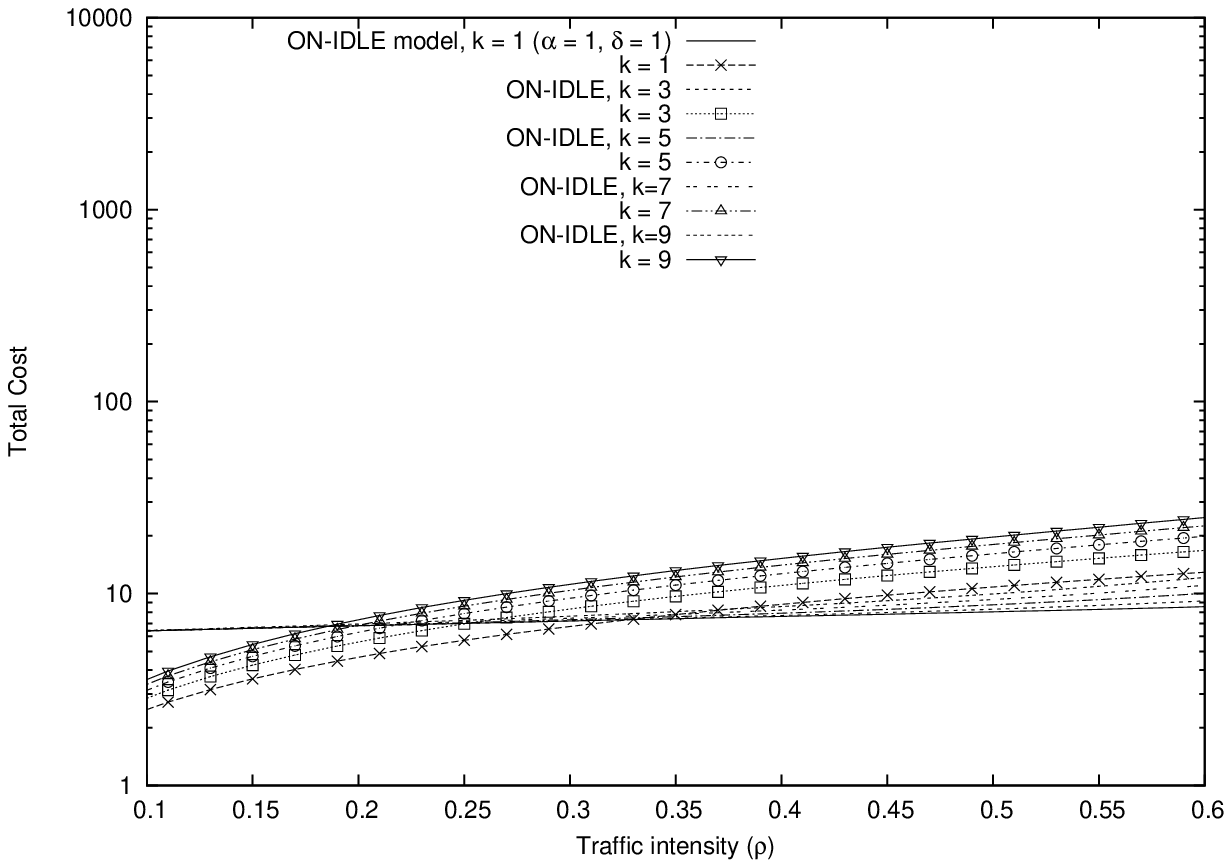}
\caption{Total Cost $\alpha = 1, \delta = 1$.}
\label{power_qos_tradeoff_vs_rho_ap1delta1:fig}
\end{center}
\end{minipage}
\end{tabular}
\end{figure}

\begin{figure}
\begin{center}
\includegraphics[scale=0.65]{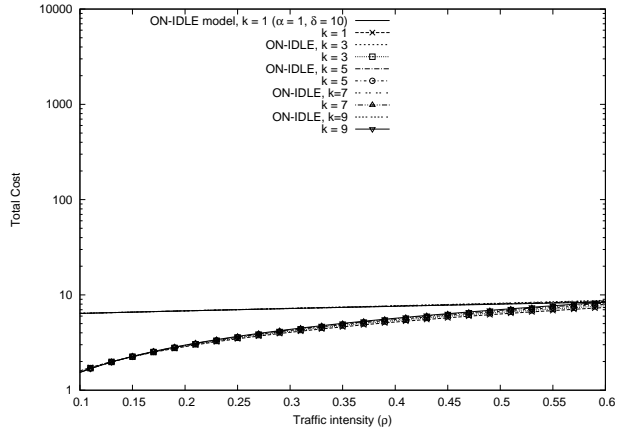}
\end{center}
\caption{Total Cost $\alpha = 1, \delta = 10$.}
\label{power_qos_tradeoff_vs_rho_ap1delta10:fig}
\end{figure}

\section{Concluding remarks}\label{concluding_remark:sec}
In this paper, we have considered the M${}^{\rm X}$/M/$c$ queueing system with staggered setup where only one server can be in setup mode at a time. A server is turned off immediately after serving a job and there is no waiting customer. If there are some waiting customers, OFF servers are turned on one by one. Using a generating function approach, we have obtained the partial generating functions of the queue length. We also have obtained recursive formulae for computing the factorial moments of the number of waiting jobs. Numerical experiments have shown some insights into the performance of the system. Furthermore, it is also important to consider the case where a fixed number of servers are always kept ON in order to reduce the delay of customers. It is also interesting to find the relation between the decomposition formula in this paper with that of Fuhrmann and Cooper~\cite{Fuhrmann85}. We have obtained partial generating functions for the joint queue lengths. A possible future work may be to obtain the tail asymptotic for the joint queue lengths.

\section*{Acknowledgements}
The author would like to thank Professor Herwig Bruneel of Ghent University and Professor Onno Boxma of Eindhoven University of Technology for useful remarks on the conditional decomposition. 
This research was supported in part by Japan Society for the Promotion of Science, JSPS Grant-in-Aid for Young Scientists (B), Grant Number 2673001. 



\begin{thebibliography}{99}

\bibitem{Adan06} Adan, I. J., Van Leeuwaarden, J. S. H., and Winands, E. M. (2006). On the application of Rouche's theorem in queueing theory. Operations Research Letters, 34, 355-360.

\bibitem{Artalejo05}
Artalejo, J. R., Economou, A. and  Lopez-Herrero, M. J. (2005). Analysis of a multiserver queue with setup times. Queueing Systems, 51, 53-76.

\bibitem{Barroso07}	
Barroso, L. A. and Holzle, U. (2007). The case for energy-proportional computing. Computer, 40, 33-37.

\bibitem{Burke}
Burke, P. J. (1975). Delays in Single-Server Queues with Batch Input. Operations Research, 23, 830-833.

\bibitem{Chen05}
  Chen, Y., Das, A., Qin, W., Sivasubramaniam, A., Wang, Q., and Gautam, N. (2005). Managing server energy and operational costs in hosting centers. ACM SIGMETRICS Performance Evaluation Review 33, 303-314.

\bibitem{Choudhury98}	
Choudhury, G. (1998). On a batch arrival Poisson queue with a random setup time and vacation period, Computers and Operations Research, 25, 1013-1026.

\bibitem{Choudhury00}	
Choudhury, G. (2000). An $M^X$/G/1 queueing system with a setup period and a vacation period, Queueing Systems, 36, 23-38.

\bibitem{Cromie}
Cromie, M. V., Chaudhry, M. L., Grassmann, W. K. (1979). Further results for the queueing system $M^X$/M/c. Journal of the Operational Research Society, 755-763.

\bibitem{Dean08} 
Dean, J. and Ghemawat, S. (2008). MapReduce: simplified data processing on large clusters. Communications of the ACM, 51(1), 107-113.

\bibitem{Fuhrmann85} 
Fuhrmann, S. W. and Cooper, R. B. (1985). Stochastic decompositions in the M/G/1 queue with generalized vacations. Operations research, 33(5), 1117-1129.








\bibitem{Gandhi10a}
 Gandhi, A., Harchol-Balter, M. and Adan, I. (2010). Decomposition results for an M/M/k with staggered setup. ACM SIGMETRICS Performance Evaluation Review, 38, 48-50.


\bibitem{Gandhi10}
	\newblock{Gandhi, A, Harchol-Balter, M. and Adan, I.  (2010).}
	\newblock{Server farms with setup costs.}
    \newblock{Performance Evaluation, 67, 1123--1138.}




\bibitem{Gandhi10b}
Gandhi, A., Gupta, V., Harchol-Balter, M. and  Kozuch, M. A. (2010). Optimality analysis of energy-performance trade-off for server farm management. Performance Evaluation, 67, 1155-1171.


\bibitem{gandhi11}
 Gandhi, A., and Harchol-Balter, M. (2011). How data center size impacts the effectiveness of dynamic power management. In Proceedings of 49th Annual Allerton Conference on Communication, Control, and Computing (Allerton), 1164-1169.


\bibitem{Gandhi11b}    
Gandhi, A., Harchol-Balter, M. and Kozuch, M. A. (2011). The case for sleep states in servers. In Proceedings of the 4th Workshop on Power-Aware Computing and Systems, article no. 2.



\bibitem{Gandhi13}
Gandhi, A. and Harchol-Balter, M. (2013). M/G/$k$ with staggered setup. Operations Research Letters, 41, 317--320.



\bibitem{Greenberg08}
Greenberg, A., Hamilton, J., Maltz, D. A. and Patel, P. (2008). The cost of a cloud: research problems in data center networks. ACM SIGCOMM Computer Communication Review, 39, 68-73.



\bibitem{Mazzucco12}
Mazzucco, M. and Dyachuk, D. (2012). Balancing electricity bill and performance in server farms with setup costs. Future Generation Computer Systems, 28, 415-426.


\bibitem{Meisner09}
Meisner, D., Gold, B. T. and Wenisch, T. F. (2009). PowerNap: eliminating server idle power. ACM Sigplan Notices, 44, 205-216.

\bibitem{Mitrani11}
	\newblock{Mitrani, I. (2011).}
	\newblock{Service center trade-offs between customer impatience and power consumption.}
    \newblock{Performance Evaluation, 68,  1222--1231.}
    
\bibitem{Mitrani13}    
    Mitrani, I. (2013). Trading power consumption against performance by reserving blocks of servers. Computer Performance Engineering. Springer Berlin Heidelberg. 1-15.

    
 \bibitem{Mitrani13anor}
    Mitrani, I. (2013). Managing performance and power consumption in a server farm. Annals of Operations Research, 202, 121-134.


\bibitem{phungduc14} 
Phung-Duc, T. (2014). Impatient customers in power-saving data centers. In Proceedings of 21th International Conference on Analytical and Stochastic Modeling Techniques and Applications (ASMTA 2014), Lecture Notes in Computer Science LNCS 8499, Springer,185--199.



\bibitem{phungduc14b} 
Phung-Duc, T. (2014). Server Farms with Batch Arrival and Staggered Setup. In Proceedings of The Fifth Symposium on Information and Communication Technology (SoICT). ACM pp. 240--247, 2014.


\bibitem{phungduc14c} 
Phung-Duc, T. (2014). Exact Solutions for M/M/c/Setup Queues. Preprint: http://arxiv.org/abs/1406.3084.

   
\bibitem{Schwartz12}
Schwartz, C., Pries, R. and Tran-Gia, P. (2012). A queuing analysis of an energy-saving mechanism in data centers. In Proceedings of International Conference on Information Networking (ICOIN), 70-75. 

\bibitem{Takagi90}
Takagi, H. (1990), Priority queues with setup times, Operations Research, 38, 667-677.

\bibitem{Tian99}
Tian N., Li Q. L. and Gao J. (1999). Conditional stochastic decompositions in the M/M/$c$ queue with server vacations. Stochastic Models, 15, 367-377.


\bibitem{Bischof01}
 Wolfgang B. (2001). Analysis of M/G/1-queues with setup times and vacations under six different service disciplines, Queueing Systems, 39, 265-301. 

\bibitem{Tian03a}
Zhang Z.~G. and Tian N. (2003). Analysis on queueing systems with synchronous vacations of partial servers. Performance Evaluation, 52, 269-282.

\bibitem{Tian03b}
Zhang, Z.~G. and Tian N. (2003). Analysis of queueing systems with synchronous single vacation for some servers. Queueing Systems, 45, 161-175.


\end{thebibliography}
\end{document}